\documentclass[10pt]{article}

\usepackage{amsthm,amssymb,stmaryrd}

\usepackage{amsthm,amssymb,stmaryrd}
\usepackage{amsfonts,tikz}
\usepackage{mathtools}
\usepackage{dsfont}
\usepackage{epsfig}
\usepackage{mathrsfs}
\usepackage{amsmath}
\usepackage{cases}
\usepackage{graphicx}

\def\mid{{\,|\,}}

\def\res{{\rm res}}

\let\set\mathbbm

\def\ord{\mbox{ord }}

\def\ord{{\rm ord}}

\DeclareMathOperator{\lclm}{lclm}
\DeclareMathOperator{\dres}{dres}

\newcommand{\bN} {\mathbb{N}}
\newcommand{\bC} {\mathbb{C}}
\newcommand{\bQ} {\mathbb{Q}}
\newcommand{\bZ} {\mathbb{Z}}

\newtheorem{thm}{Theorem}[section]

\newtheorem{lem}[thm]{Lemma}

\newtheorem{definition}[thm]{Definition}

\newtheorem{exam}[thm]{Example}

\makeatletter \@addtoreset{equation}{section}

\def\BibTeX{{\rm B\kern-.05em{\sc i\kern-.025em b}\kern-.08emT\kern-.1667em\lower.7ex\hbox{E}\kern-.125emX}}

\begin{document}

%
\title{Non-minimality of Minimal Telescopers \\Explained by Residues
\thanks{S.\ Chen, X.\ Li, and Y.\ Wang were partially supported by the National Key R\&D Program of China (No.\ 2020YFA0712300 and No.\ 2023YFA1009401), the NSFC grant (No.\ 12271511), the CAS Project for Young Scientists
in Basic Research (No.\ YSBR-034), the CAS Funds of the Youth Innovation Promotion Association (No.\ Y2022001), and the Strategic Priority Research Program of the Chinese Academy of Sciences (No.\ XDB0510201).
M.\ Kauers, C.\ Koutschan, X.\ Li, and Y.\ Wang were supported by the Austrian FWF grant 10.55776/I6130. M.\ Kauers was also supported by the Austrian FWF grants 10.55776/PAT8258123 and 10.55776/PAT9952223. R.-H.\ Wang was supported by the NSFC grant (No.\ 12101449). 
All authors also were supported by the International Partnership Program of Chinese Academy of Sciences (Grant No.\ 167GJHZ2023001FN).
}}

\author{
Shaoshi Chen$^{a, b}$, Manuel Kauers$^{c}$, Christoph Koutschan$^{d}$\\
\bigskip
Xiuyun Li$^{a, b, c}$, Rong-Hua Wang$^{e}$, and Yisen Wang$^{a,b,d}$\\
$^a$KLMM,\, Academy of Mathematics and Systems Science, \\ Chinese Academy of Sciences, Beijing, 100190, (China)\\
$^b$School of Mathematical Sciences, \\University of Chinese Academy of Sciences,\\ Beijing 100049, (China)\\
$^c$Institute for Algebra, Johannes Kepler University\\ Linz A-4040, (Austria)\\
$^d$RICAM, Austrian Academy of Sciences\\ Linz A-4040, (Austria)\\
$^e$School of Mathematical Sciences, Tiangong University\\  Tianjin, 300387, (China)\\
{\sf schen@amss.ac.cn, manuel.kauers@jku.at}\\
{\sf lixiuyun@amss.ac.cn, christoph.koutschan@oeaw.ac.at}\\
{\sf  wangronghua@tiangong.edu.cn,  wangyisen@amss.ac.cn}
}

\maketitle

\begin{abstract}
Elaborating on an approach recently proposed by Mark van Hoeij, we 
continue to investigate why creative telescoping occasionally fails
to find the minimal-order annihilating operator of a given definite
sum or integral. We offer an explanation based on the consideration
of residues. 
\end{abstract}


%
%
\maketitle

\def\id{\operatorname{id}}
\let\set\mathbb

\section{Introduction}

Creative telescoping is the standard approach to definite summation and integration in computer algebra.
Its purpose is to find an annihilating operator for a given definite sum $\sum_k f(n,k)$ or a given definite
integral $\int_\Omega f(x,y)dy$.

Such operators are obtained from annihilating operators of the summand or integrand that have a particular
form. In the case of summation, suppose that we have
\begin{alignat}1\label{eq:telescoper1}
(L - (S_k-1)Q)\cdot f(n,k) = 0
\end{alignat}
for some operator $L$ that only involves $n$ and the shift operator $S_n$ but neither $k$ nor the shift operator~$S_k$,
and another operator $Q$ that may involve any of $n,k,S_n,S_k$.
Summing the equation over all $k$ yields
\[
  L\cdot\sum_k f(n,k) = \bigl[Q\cdot f(n,k)\bigr]_{k=-\infty}^\infty.
\]
If the right-hand side happens to be zero, we find that $L$ is an annihilating operator for the sum. 

In the case of integration, having
\begin{alignat}1\label{eq:telescoper2}
(L - D_yQ)\cdot f(x,y)=0
\end{alignat}
for some operator $L$ that only involves $x$ and the derivation $D_x$ but neither $y$ nor the derivation~$D_y$,
and some other operator $Q$ that may involve any of $x,y,D_x,D_y$, implies the equation
\[
 L\cdot\int_\Omega f(x,y)\,dy = \bigl[Q\cdot f(x,y)\bigr]_{\Omega}.
\]
If the right-hand side happens to be zero, we find that $L$ is an annihilating operator for the integral.

An operator $L$ as in equations \eqref{eq:telescoper1} and~\eqref{eq:telescoper2} is called 
a \emph{telescoper} for~$f$, and $Q$ is called a \emph{certificate} for~$L$. The degree of $S_n$ or $D_x$ in $L$ is called the order of~$L$.
If $L$ is such that there is no telescoper of lower order, then $L$ is called a \emph{minimal telescoper}.
The minimal telescoper is unique up to multiplication by rational functions (from the left).

Algorithms for testing the existence of telescopers and computing them if they exist meanwhile have a long history 
in computer algebra, see~\cite{abramov02a,abramov03,kauers23,PWZbook1996,Zeilberger1990c,Zeilberger1991}
for classical results and recent developments on the matter.
In his recent paper~\cite{VANHOEIJ2025102342}, van Hoeij proposed a fresh view on creative telescoping.
He explains why a telescoper can often be written as a least common left multiple of
smaller operators, and why the minimal telescoper is sometimes not the minimal-order annihilating 
operator for the sum or integral under consideration.

Let $C$ be a field of characteristic zero and $C(n, k)$ be the field of rational functions in $n, k$ over~$C$. Let $A_{n,k}=C(n,k)\langle S_n, S_k\rangle$ be the ring
of all linear recurrence operators in $S_n, S_k$ with rational function coefficients, and $A_n=C(n)\langle S_n\rangle$ be the subalgebra consisting of all operators that do not involve $k$ or~$S_k$. For a given summand $f(n,k)$, consider the $A_n$-module $\Omega:= A_{n,k}\cdot f(n,k)$
and the quotient module $M:=\Omega/((S_k-1)\Omega)$. An operator $L\in A_n$ is then a
telescoper for $H=f(n,k)$ if and only if it is an annihilating operator of the image $\overline H$ of
$H$ in~$M$.

In this setting, van Hoeij makes the following observations:
\begin{itemize}
\item If $M$ can be written as a direct sum of submodules, say $M=M_1\oplus M_2$, 
  then the minimal telescoper for $H$ is the least common left multiple of the
  minimal annihilating operators of the projections $\pi_1(H)$ and $\pi_2(H)$ of $H$
  in $M_1$ and~$M_2$, respectively.
\item If, moreover, the definite sum whose summand corresponds to $\pi_1(H)$ happens
  to be zero identically, then every annihilating operator of $\pi_2(H)$ is already
  an annihilating operator of the definite sum over~$H$, even though it may not be
  a telescoper for~$H$.
\end{itemize}
In order to take advantage of the second observation, it is necessary to understand
under which circumstances a definite sum can be zero. Such ``vanishing sums'' are 
themselves examples when a minimal telescoper fails to be a minimal annihilator.
For example, we have $\sum_k (-1)^k\binom {2n+1}k^2=0$, so the minimal annihilator is~$1$.
However, the minimal telescoper of $(-1)^k\binom{2n+1}k^2$ is $L=(2n+3)S_n+(8n+8)$.
Note that since $L$ is irreducible, the module~$M$, which is isomorphic to $A_n/\langle L\rangle$, has no
nontrivial submodules.

We propose an explanation of why certain sums are identically zero which is based
on the investigation of residues. Also based on residues, we will explain why telescopers tend to be 
least common left multiples. 
We are not the first to use residues in the context of creative telescoping.
For rational functions and algebraic functions in the differential case, it was observed
by Chen, Kauers, and Singer~\cite{Chen2012} that telescopers and residues are closely related.
Chen and Singer also used residues in the context of summation problems~\cite{ChenSinger2012}.
Residues are also tied to creative telescoping through the equivalence of extracting
residues with taking diagonals and positive parts and the computation of Hadamard
products~\cite{bostan16b}.
We are also not the first to study the non-minimality of telescopers. Besides van Hoeij's recent work~\cite{VANHOEIJ2025102342}, 
the problem was investigated by Paule~\cite{Paule1994} who proposed the method of creative symmetrizing. This method was further developed~\cite{PauleRiese1997,PauleSchorn1995} and enhanced by incorporating contiguous relations~\cite{Paule2021}.
An approach reminiscent to van Hoeij's ideas already appeared in a technical report by Chyzak~\cite{Chyzak2000ANM}. By translating multiple binomial sums to rational integrals, Bostan, Lairez, and Salvy~\cite{BostanLairezSalvy16} approach the non-minimality problem by a technique they call geometric reduction.

\section{Residues and Telescopers for Rational Functions} \label{SECT:residuerat}
Residues have played an important role in rational integration and summation~\cite{Arreche2024,Arreche2022,BronsteinBook,ChenSinger2012, Matusevich2000}. In this section, we will first use residues in the continuous setting to explain why minimal telescopers may not always lead to minimal annihilators for integrals and then use residues in the discrete setting to explain some vanishing sums. 

\subsection{The integration case}\label{SUBSECT:ratint}

Let $F=C(x)$, so that the bivariate rational function field $C(x,y)$ can be viewed
as a univariate rational function field~$F(y)$.
An element $f$ of $F(y)$ is said to be \emph{integrable in $F(y)$} if $f = D_y(g)$ for some $g\in F(y)$. 

Any rational function $f=a/b\in F(y)$ with $a, b\in F[y]$ and $\gcd(a, b)=1$ can be uniquely written as
\[f = p + \sum_{i=1}^n \sum_{j=1}^{m_i} \frac{\alpha_{i,j}}{(y-\beta_i)^j},\]
where $p\in F[y], n, m_i\in \bN, \alpha_{i, j}, \beta_i\in \overline{F}$, and the $\beta_i$'s are distinct roots of $b$. Note that all the $\alpha_{i, j}$'s are in the field $F(\beta_1, \ldots, \beta_n)$. 
The value $\alpha_{i, 1}\in \overline{F}$ is called the \emph{residue} (in~$y$) of $f$ at~$\beta_i$, denoted by $\res_y(f, \beta_i)$. Let $P, Q\in F[y]$ be such that $\gcd(P, Q)=1$ and $Q$ is squarefree and let $\beta\in \overline{F}$ be a zero of $Q$. Then we have Lagrange's residue formula
\[\res_y\left(\frac{P}{Q}, \beta\right) = \frac{P(\beta)}{D_y(Q)(\beta)}.\]

It is well-known that a rational function is integrable in $F(y)$ if and only if all its residues in $y$ are zero (see~\cite[Proposition~2.2]{ChenSinger2012}). So residues are the obstruction to the integrability in $F(y)$. From this fact and the commutativity between the derivation in $x$ and taking the residue in~$y$, we have that the minimal telescoper of a rational function in $C(x, y)$ is the least common left multiple of the minimal annihilating operators of its residues in $y$ which are algebraic functions in~$\overline{C(x)}$ (see~\cite[Theorem~6]{Chen2012}).

Now consider the integral
\[
  I(x):= \int_{-\infty}^\infty f(x, y)\, dy \quad \text{with $f := \frac1{y^4+xy^2+1}$ and $x>2$}.
\]
We have $I(x)=\pi/\sqrt{x+2}$, so the integral has the minimal annihilator
$(2x+4)D_x+1$. The minimal telescoper for $f$ however is $L=(4x^2-16)D_x^2+12xD_x+3$.
Let us see why the minimal telescoper overshoots in this example.

Let $\alpha,\beta\in\overline{\set Q(x)}$ be such that $\alpha,-\alpha,\beta,-\beta$
are the poles of~$f$ and $\beta=\alpha(\alpha^2+x)$. Then we have the residues
\[
\res_{y}(f, \pm\alpha) = \pm\frac{\alpha(2-x^2-\alpha^2x)}{2(x^2-4)}
\,\, \text{and}\,\,  
\res_{y}(f, \pm\beta) = \pm\frac{\alpha(2\alpha^2+x)}{2(x^2-4)}.
\]

Note that each of the four residues has the telescoper $L$ as its minimal annihilator.
This does not explain yet why the telescoper factors and overshoots. To explain
this, we need to observe that the sum $\res_y(f, \alpha)+\res_y(f, \beta)$
is annihilated by $(2x+4)D_x+1$. By the residue theorem, the sum of these residues
is equal (up to a multiplicative constant) to the following contour integral:
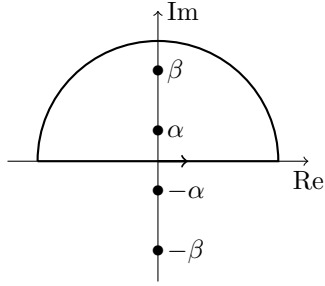
\begin{center}
  \begin{tikzpicture}[scale=.4]
    \draw[->](-5,0)--(5,0) node[below]{Re};
    \draw[->](0,-4)--(0,5) node[right]{Im};
    \draw
    (0,1)node{$\bullet$}node[right]{$\alpha$}
    (0,-1)node{$\bullet$}node[right]{$-\alpha$}
    (0,3)node{$\bullet$}node[right]{$\beta$}
    (0,-3)node{$\bullet$}node[right]{$-\beta$};
    \draw[thick] (-4,0)--(4,0) arc (0:180:4);
    \draw[thick,->] (0,0)--(1,0);
  \end{tikzpicture}
\end{center}
By increasing the contour indefinitely, we see that it is also the value of the real
integral $I(x)=\pi/\sqrt{x+2}$. 
As creative telescoping does not know the contour but only the integrand, it must return
a telescoper that works for every contour, in particular one that encircles only one of the
poles. For such a contour, the minimal telescoper is indeed the minimal annihilator.

In van Hoeij's language of submodules, translated to the differential case, consider
$\Omega=C(x,y)$, $M=\Omega/D_y\Omega$, and $A_x=C(x)\langle D_x\rangle$. The submodule 
generated by $f$ in $M$ is $N=\operatorname{span}_{C(x)}(f+D_y\Omega,y^2f+D_y\Omega)$. 
Note that $\dim_{C(x)}N=\ord(L)=2$. The module $N$ admits a decomposition
$N=N_+\oplus N_-$ where $N_{\pm} = \operatorname{span}_{C(x)}((1\pm y^2)f + D_y\Omega)$,
which suggests writing
\[
  f = \frac{1 + y^2}2 f + \frac{1 - y^2}2 f.
\]
Indeed, the minimal telescoper of $\frac{1+y^2}2 f$ is $(2x+4)D_x+1$, the minimal
telescoper of $\frac{1-y^2}2f$ is $(2x-4)D_x+1$, and $L$ is the least common left
multiple of these operators. Because of 
\[
\res_y(f, \alpha) = \res_y(y^2 f,\beta)
\quad\text{and}\quad
\res_y(f, \beta) = \res_y(y^2 f, \alpha),
\]
the residues of $\frac{1-y^2}2f$ at $\alpha$ and $\beta$ cancel each other, so
\[
  \int_{-\infty}^\infty \frac{1-y^2}2f\,dy=0,
\]
and that's why the factor $(2x-4)D_x+1$ of $L$ is not needed for~$I(x)$.

\subsection{The summation case} \label{SUBSECT:ratsum}
As a discrete analogue of residues for rational integration, discrete residues are introduced to study the summability problem and the existence problem of telescopers for rational functions in~\cite{ChenSinger2012}. Efficient algorithms for computing discrete residues and their variants are given in~\cite{Arreche2024,Arreche2022,ArrecheZhang2024}. 

Let $S_x$ and $S_y$ denote the usual shift operators of $C(x, y)$ with respect to $x$ and $y$, respectively. Let $\Delta_y$ denote the difference operator defined by $\Delta_y(r) = S_y(r)-r$ for any $r\in F(y)$.
A rational function $f\in F(y)$ is said to be \emph{summable in $F(y)$} if $f = \Delta_y(g)$ for some $g\in F(y)$.  For any elements $\beta\in \overline{F}$, we call the set $\{\beta + i\mid i\in \bZ\}$ a \emph{$\bZ$-orbit} of $\beta$ in $\overline{F}$, denoted by $[\beta]_{\bZ}$. Any rational function $f \in F(y)$ can be decomposed into the form
 \[f = p + \sum_{i=1}^n \sum_{j=1}^{m_i} \sum_{\ell=0}^{d_{i,j}} \frac{\alpha_{i,j, \ell}}{(y-(\beta_i +\ell))^j},\]
 where $p \in F[y]$, $m,n_i,d_{i,j} \in \bN$, $\alpha_{i,j,\ell}, \beta_i \in \overline{F}$, and the $\beta_i$'s are in distinct $\bZ$-orbits. The sum $\sum_{\ell=0}^{d_{i,j}} \alpha_{i,j,\ell}$ is called the \emph{discrete residue in $y$} of $f$ at the $\bZ$-orbit $[\beta_i]_{\bZ}$ of multiplicity $j$, denoted by $\dres_y(f,[\beta_i]_{\bZ},j)$. 
By Proposition 2.5 in~\cite{ChenSinger2012}, discrete residues are the precise obstruction for rational functions to be summable, i.e.,  $f \in F(y)$ is summable in $F(y)$ if and only if all of the discrete residues of $f$ are zero. 

We recall a very old result due to Nicole~\cite{Tweedie1917} that describes a family of summable rational functions and then use this result to explain some vanishing sums.
The idea behind this theorem has become part of the classical summation folklore and also explained, for example, in Section~5.3 of~\cite{GKP1994}.

\begin{lem}[Nicole, 1717]\label{THM:nicole}
Let $n \geq 2$ be an integer and $P\in F[y]$ be such that~$\deg_y(P) \leq n-2$. Then the rational function
\[f = \frac{P(y)}{(y+\beta_1)\cdots (y+\beta_n)} \]
is summable in $F(y)$ for all $\beta_i\in \overline{F}$ with $\beta_i-\beta_j\in \bZ\setminus\{0\}$ for $i\neq j$.
\end{lem}
\begin{proof}
By partial fraction decomposition, we get
\begin{equation}\label{EQ:sumf}
f = \sum_{i=1}^n \frac{\alpha_i}{y+\beta_i}, \quad \text{where $\alpha_i\in \overline{F}$.}
\end{equation}
Note that the $\beta_i$'s are in the same $\bZ$-orbit. By Proposition 2.5 in~\cite{ChenSinger2012},
$f$ is summable in $F(y)$ if and only if the sum $\sum_{i=1}^n \alpha_i$ is zero. 
By normalizing $f$ in~\eqref{EQ:sumf}, we get
\[P = (\alpha_1+\cdots+ \alpha_n) y^{n-1} + \text{terms with degree lower than $n-1$}.\]
Since~$\deg_y(P) \leq n-2$, it holds that~$\sum_{i=1}^{n} \alpha_i =0$. 
\end{proof}

When $F$ is the field of complex numbers,  the
identity $\sum_{i=1}^{n} \alpha_i =0$ also follows from Cauchy's residue 
theorem since the residue of $f$ at infinity is zero. 

As a corollary of Nicole's lemma, we obtain a class of vanishing sums.
For any polynomial~$P\in F[y]$ with $\deg_y(P)\leq n-1$, we consider the rational function
\[f = \frac{P(y)}{ y(y+1)\cdots (y+n)} = \sum_{k=0}^{n} \frac{\alpha_k}{y+k},\]
which is summable in $F(y)$ by Nicole's lemma. Since the denominator of $f$ is squarefree, Lagrange's residue formula implies that
\[\alpha_k = \frac{(-1)^k P(-k)}{ k! (n-k)!}.\]
Then we have the vanishing sum
\[\sum_{k=0}^n \frac{(-1)^k P(-k)}{ k! (n-k)!} = 0.\]
\begin{exam}\label{EX:euler} 
To show the combinatorial identity
\[\sum_{k=0}^n (-1)^k \binom{n}{k}k^{j}  = 0, \quad \text{where $ n\geq 2$ and $0 \leq j < n$,} \]
we consider the rational function
\[f =\frac{P}{Q} =  \frac{n!(-y)^j}{y(y+1)\cdots (y+n)} = \sum_{k=0}^n \frac{\alpha_k}{y+k}.\]
By Lagrange's residue formula, we have 
\[\alpha_k = (-1)^k \binom{n}{k}k^{j}. \]
Since $0 \leq j < n$, we have $\deg_y(P) \leq \deg_y(Q) -2$. Then the identity $\sum_{k=0}^{n} \alpha_k = 0$ holds. 
\end{exam}

\begin{exam}\label{EX:vanisum}
To show the combinatorial identity
\[\sum_{k=0}^n \binom{2k}{k} \binom{2n-2k}{n-k} \frac{1}{2k-1} = 0, \quad \text{where $n\geq 1$},\]
we consider the rational function
\[f =\frac{P}{Q} = -\frac{2^n \prod_{i= 1}^{n-1} (2(y+i)+1)}{y(y+1)\cdots (y+n)} = \sum_{k=0}^n \frac{\alpha_k}{y+k}.\]
 By Lagrange's residue formula, we get
\[\alpha_k = \binom{2k}{k} \binom{2n-2k}{n-k} \frac{1}{2k-1}.\]
Since $\deg_y(P) = n-1$ and  $\deg_y(Q) = n+1$,  Nicole's lemma implies the identity $\sum_{k=0}^{n} \alpha_k = 0$.
\end{exam}
We will see more applications of Nicole's lemma in Section~\ref{SUBSECT:zerosum}.

\section{Residual Forms and Prescopers for Hypergeometric Terms} \label{SECT:residuehyper}

We now focus on creative telescoping for hypergeometric terms. We will use residual forms introduced in~\cite{ChenHuangKauersLi}
to construct submodules in order to find right factors of minimal telescopers and then investigate the automorphisms 
and the non-minimality phenomenon of minimal telescopers for hypergeometric sums.  These studies continue the development of
the submodule approach initialized by van Hoeij~\cite{VANHOEIJ2025102342}.

To be more compatible with the customary usage, we will now use $n$ and $k$ instead of $x$ and $y$, respectively.
A sequence $H(n, k)$ is called a~\emph{hypergeometric term} over $C(n,k)$ with respect to $n$ and $k$ if the two shift quotients 
$S_n(H)/H$ and $S_k(H)/H$
are rational functions in $C(n,k)$. A hypergeometric term $H$ is said to be~\emph{hypergeometric summable} in $k$ if $H = \Delta_k(G)$
for some hypergeometric term $G$.
A nonzero linear operator $L \in C(n)\langle S_n \rangle $ is called a \emph{telescoper} for $H$ if there exists
another hypergeometric term $G(n, k)$ such that
\begin{equation} \label{EQ:telescoping}
L(H(n,k))=\Delta_k(G(n,k)).
\end{equation}
Recall that~$p\in C(n)[k]$ is \emph{shift-free} in $k$ if~$\gcd(p,S_k^{i}(p))=1$ for all~$i \in \bZ \setminus \{0\}$.
A rational function $f=a/b \in C(n, k)$ is~\emph{shift-reduced} 
in $k$ if $\gcd(a, S_k^{i}(b)) = 1$ for all $i \in \bZ$. A nonzero polynomial $p \in C(n)[k]$ is~\emph{strongly prime} 
with a rational function~$f = a/b$ if $\gcd(p, S_k^{-i}(a))=\gcd(p, S_k^{i}(b))=1$ for all $i \in \bN$.
By computing rational normal forms as in~\cite{AbramovPetkovsek2002b}, one can 
write $f\in C(n, k)$ as
\begin{equation}\label{EQ:ratNF}
  f =  \frac{S_k(S)}{S}\cdot K,
\end{equation}
where $S, K \in C(n, k)$ such that $K$ is shift-reduced in~$k$.
The rational functions $K$ and $S$ are called \emph{kernel} and \emph{shell} of~$f$, respectively. 
Let $f = S_k(H)/H$. Then $H = S\cdot H_0$ with $S_k(H_0)/H_0 = K$.  Write $K =u/v$ with $u, v\in C(n)[k]$ and $\gcd(u, v)=1$. 
Let $\phi_K\colon C(n)[k] \rightarrow C(n)[k]$ be a $C(n)$-linear map defined by 
\[\phi_K(p) = uS_k(p) - vp\quad \text{for all $p \in C(n)[k]$}.\]
Let $W_K$ be the standard complement of the image $\text{im}(\phi_K)$ in $C(n)[k]$ such that
$C(n)[k] = \text{im}(\phi_K) \oplus W_K$. 
By the modified Abramov--Petkov\v{s}ek reduction~\cite{ChenHuangKauersLi}
we can decompose $H$ into 
\begin{equation}\label{EQ:adddecomp}
  H = \Delta_k(r \cdot H_0) + \left(\frac{a}{b} + \frac{p}{v}\right) H_0
\end{equation}
where $r\in C(n, k), p\in W_K$, and $a, b\in C(n)[k]$ such that $\deg_k(a)<\deg_k(b)$, $\gcd(a,b)=1$, and $b$ is shift-free in $k$ and strongly prime with $K$.
By Proposition~4.7 and Theorem~4.8 in~\cite{ChenHuangKauersLi}, we have $W_K$ is finite-dimensional over $C(n)$
and $H$ is hypergeometric summable in $k$ if and only if $a = 0$ and $p=0$. So the form $(a/b+p/v)H_0$ is the obstruction
to the hypergeometric summability.  For this reason, we call $(a/b+p/v)H_0$ a \emph{residual form} of $H$ with respect to $\Delta_k$.

Let $\Omega$ be the $A_{n}$-module $C(n, k)\cdot H$.
Note that $\Delta_k(\Omega)$ is an $A_n$-submodule of $\Omega$. Let $M$ denote the quotient module $\Omega/\Delta_k(\Omega)$. An operator $L\in A_n$
is a telescoper for $H$ if and only if $L$ is an annihilator of the image $\overline{H}$ of $H$ in~$M$.

\begin{lem}\label{THM:submoduleN}
    Let $H_0$ and $v$ be defined as in~\eqref{EQ:adddecomp} and let
    \[N := \left\lbrace \frac{p}{v} H_0 + \Delta_k(\Omega) \ \middle|\ p \in W_K  \right\rbrace.\]
    Then $N$ is an $A_n$-submodule of~$M$.
\end{lem}
    \begin{proof}
    By~\cite[Proposition~5.2]{Huang2016} with~$b_0 =1$, for any~$i \in \bN$,
    \[S_n^i\left(\frac{p}{v}H_0\right) \equiv \frac{p_i}{v} H_0 \bmod \Delta_k(\Omega)\]
    for some~$p_i \in W_K$. The lemma follows.
    \end{proof}
Note that $N$ is independent of the choice of~$S$ and~$K$ in the rational normal form \eqref{EQ:ratNF}. We will call $N$ a \emph{kernel submodule} of $M$ which is an $A_n$-submodule and a finite-dimensional vector space over $C(n)$. 
Recall that an operator $L$ is a telescoper for $H$ if it annihilates $\overline{H}$ in~$M$.
Therefore, if $N$ is any submodule of $M$, then for an operator $L$ to be a telescoper, it 
is \emph{necessary} that $L$ maps $\overline{H}$ into~$N$, although this condition is in
general \emph{not sufficient} for being a telescoper.
This observation motivates the following definition of prescopers for hypergeometric terms.
An analogous definition was introduced in~\cite[Section~6.2]{GeddesLeLi2004} for hyperexponential 
functions.
\begin{definition}\label{DEF:prescoper}
    A nonzero operator $R\in C(n)\langle S_n \rangle$ is called a \emph{prescoper} 
    for $H$ with respect to $k$ if~$R(H) + \Delta_k(\Omega) \in N$, i.e., there exists $p \in W_K$ such that 
    \[R(H) \equiv \frac{p}{v} H_0 \bmod \Delta_k(\Omega).  \]
A prescoper is said to be minimal if it has minimal degree in~$S_n$. 
\end{definition}
By definition, it is clear that telescopers are prescopers.  The next lemma shows that the minimal prescoper 
for~$H$ is a right factor of the minimal telescoper for~$H$ if they exist. 
\begin{lem}\label{LEM:rightfactor}
    Let~$N\subseteq M$ be $A_n$-modules and~$m \in M$. Suppose that $R\in A_n$ is the minimal annihilator for~$m+N \in M /N$ and $T$ is the minimal annihilator for~$R(m)$, then $T \cdot R$ is the minimal annihilator for~$m \in M$.
    \end{lem}
 \begin{proof}
We firstly observe that $T \cdot R$ is an annihilator for~$m \in M$. 
Let $L$ be any annihilator for~$m$. Then $L$ must be an annihilator for~$m+N \in M/N$, which implies that $L$ is right divisible by~$R$. 
Let~$L = \widetilde{L}\cdot R$, then $\widetilde{L}$ is an annihilator for~$R(m)$. By the minimality of~$T$, 
we have that $\widetilde{L}$ is right divisible by~$T$ and then $L$ is right divisible by~$T \cdot R$. As a consequence, we have that 
$T \cdot R$ is the minimal annihilator for~$m$.
\end{proof}

The following lemma will be used in the next sections to explore the LCLM structure of annihilators of elements in $A_n$-modules.
 
\begin{lem}\label{LEM:lclm}
Let~$M$ be an $A_n$-module and~$M = \bigoplus_{i=1}^n M_i$ be a direct-sum decomposition of $M$. For any element $m = m_1 + \cdots + m_n\in M$, 
the minimal annihilator for~$m$ is the least common left multiple of the minimal annihilators for the~$m_i$'s.
\end{lem}
\begin{proof} 
Let $L_i$ be the minimal annihilator for~$m_i \in M_i$. Suppose $L$ is an annihilator for~$m$, then 
\[L(m)=L(m_1)+\cdots+L(m_n) = 0.\]
Since~$L(m_i) \in M_i$, we have~$L(m_i)=0$, which implies that $L$ is right-divisible by~$L_i$. Thus $L$ is right-divisible by $\lclm(L_1,\ldots,L_n)$. Note that $\lclm(L_1,\ldots,L_n)$ is an annihilator for~$m$. The lemma follows.
\end{proof}

\subsection{Constructing minimal prescopers}\label{SUBSECT:direct}

We now present a method to construct minimal prescopers for
hypergeometric terms. We first recall some terminologies from~\cite[Section~4]{AbramovPetkovsek2002a} and~\cite[Section~3]{Huang2016} about properties of polynomials under shifts. Let $F$ be a field of characteristic zero. Two polynomials~$q_1, q_2 \in F[z]$ are \emph{$\sigma$-equivalent} with respect to the $F$-auto\-morphism~$\sigma$ of $F[z]$
if~$q_1 =\sigma^{j}(q_2)$ for some~$j \in \bZ\setminus\{0\}$, denoted as~$q_1 \sim_{\sigma} q_2$. 
Two shift-free polynomials $b_1, b_2 \in C(n)[k]$ are \emph{shift-related} (with respect to~$k$) if
for any nontrivial monic irreducible factor $q_1$ of~$b_1$, there exists a unique monic
irreducible factor~$q_2$ of~$b_2$ with the same multiplicity as~$q_1$ in~$b_1$ such that~$q_1$ and $q_2$ are $S_k$-equivalent and vice versa. An irreducible polynomial $p\in C[n,k] $ is~\emph{integer-linear} over $C$ if
there exist a univariate polynomial $P\in C[z]$ and a nonzero vector $ (m, \ell)\in \bZ^2 $ such that $p(n,k)=P(m n + \ell k)$. A polynomial $p\in C[n, k]$ is \emph{integer-linear} if all of its irreducible factors are integer-linear.

By the existence criterion on telescopers~\cite{Abramov2003}, a hypergeometric term $H$ as in~\eqref{EQ:adddecomp}
has a nonzero telescoper in $A_n$ if and only if $b$ is an integer-linear polynomial.
From now on, we always assume that the given hypergeometric term $H$ has a nonzero
telescoper. Since $b$ is integer-linear, shift-free in~$k$, and strongly prime with $K$, we can decompose $b$ as
\[
b = \prod_{i=1}^I \prod_{j=0}^{\ell_i -1} S_k^{\mu_{i,j}}\bigl(P_i(m_i n+ \ell_i k + j)\bigr)^{\lambda_{i, j}},
\]
where each $P_i\in C[z]$ is irreducible, $\lambda_{i, j}\in \bN$ and $m_i, \ell_i, \mu_{i,j} \in \bZ$ satisfying~$\ell_i > 0$, $\gcd(m_i, \ell_i) =1$, and $S_k^{\mu_{i,j}}\bigl(P_i(m_i n+ \ell_i k +j)\bigr)$ is strongly prime with~$K$.
Moreover, one can ensure that for all $i, i'\in \{1, \ldots, I\}$ with $i\neq i'$, at least one of the following three relations is not satisfied:
\begin{equation}\label{COND:mellP}
m_i = m_{i'} , \, \ell_i = \ell_{i'},\,  \text{and}\, \, P_i \sim_{S_z} P_{i'}.
\end{equation}
Let $\lambda_i := \max\{\lambda_{i, 0}, \ldots, \lambda_{i, \ell_i-1}\}$
 and set
 \[B_{i,j} := S_k^{\mu_{i,j}}\bigl(P_i(m_i n+ \ell_i k + j)\bigr)^{\lambda_i}.\]
Then we can write $a/b$ in the residual form of $H$ as
\begin{equation}\label{EQ:decompa/b}
    \frac{a}{b} = \sum_{i=1}^{I}\sum_{j=0}^{\ell_i -1} \frac{q_{i,j}}{B_{i,j}},
\end{equation}
where $q_{i,j} \in C(n)[k]$ such that~$\deg_k(q_{i,j})<\deg_k(B_{i,j})$. 
Let~$\hat{H} = a/b \cdot H_0$. By Definition~\ref{DEF:prescoper}, the minimal prescoper for~$H$ is equal to the minimal prescoper for~$\hat{H}$.
From the above decomposition we obtain
\[\hat H =\sum_{i=1}^{I} \hat H_i \quad \text{with}\quad \hat H_i := \sum_{j=0}^{\ell_i -1} \frac{q_{i,j}}{B_{i,j}}
\cdot H_0. \]

\begin{lem}\label{LEM:prescoperlclm} 
The minimal prescoper for $\hat H$ is the least common left multiple of the minimal prescopers for the~$\hat H_i$'s.
\end{lem}
\begin{proof}
Let $V_i \subseteq M/N$ be the set that consists of the elements
\[\sum_{j=0}^{\ell_i -1} \frac{a_{i,j}}{B_{i,j} } H_0 + N
\]
with $a_{i,j} \in C(n)[k]$ and $\deg_k(a_{i,j}) < \deg_k(B_{i,j})$. 
By~\cite[Proposition~5.4]{Huang2016}, for any~$d\in \bN$, there exist~$\widetilde{a}_{i, j} \in C(n)[k]$ with $\deg_k(\widetilde{a}_{i, j}) < \deg_k(B_{i,j})$ and $ p_d \in W_K$ such that
    \[
    S_n^d\Biggl( \sum_{j=0}^{\ell-1}\frac{a_{i, j}}{B_{i,j}} H_0\Biggr)\equiv
    \Biggl(\sum_{j=0}^{\ell-1}\frac{\widetilde{a}_{i, j}}{B_{i,j}}+\frac{p_d}{v}\Biggr)H_0 \bmod \Delta_k(\Omega).
    \] 
    This implies that $V_i$ is an
$A_n$-submodule of~$M/N$. Let~$V = \sum_{i=1}^I V_i$. Then~$\hat{H} + N$ is an element of~$V$. 
By Lemma~\ref{LEM:lclm}, it remains to show that $V= \bigoplus_{i=1}^I V_i$.
 By~\cite[Proposition~3.2]{Huang2016} the following holds: if there exist~$p_1, p_2 \in W_K$ such that
    \[\Bigl(\frac{a_1}{b_1} + \frac{p_1}{v}\Bigr) H_0 \equiv \Bigl(\frac{a_2}{b_2} + \frac{p_2}{v}\Bigr) H_0 \bmod \Delta_k(\Omega), \]
    where $b_1,b_2$ satisfy the conditions as in Equation~\eqref{EQ:adddecomp}, then $b_1$ and $b_2$ are shift-related to each other. Hence $ V_i \cap V_j = \lbrace 0 \rbrace$ for any~$i \neq j$.
\end{proof}

We next deal with the question how to compute the minimal prescoper for each~$\hat H_i$.
For each~$d \in \bN $, the modified Abramov--Petkov\v{s}ek reduction~\cite{ChenHuangKauersLi} decomposes 
\[S_n^d(\hat{H}_i) \equiv \left( r_{i,d} + \frac{p_{i,d}}{v} \right)H_0 \bmod \Delta_k(\Omega), \]
where~$p_{i,d} \in W_K$ and $r_{i,d} \in C(n,k)$, which are also contained in a finite-dimensional $C(n)$-vector space. Take the minimal~$\rho_i \in \bN$ s.t. $\sum_{d=0}^{\rho_i} e_{i,d} r_{i,d} = 0$ with~$e_{i,d} \in C(n)$ and~$e_{i,\rho_i} =1$. Then we have 
\[R_i := \sum_{d=0}^{\rho_i} e_{i,d}S_n^d \]
is the minimal prescoper for~$\hat H_i$.

For a rational function $f\in C(n, k)$ of the form
\[f = \frac{1}{(mn + \ell k)^s}, \]
where $s$ is a positive integer and $m, \ell\in \bZ$ with $\ell\neq 0$ and $\gcd(m, \ell)=1$,
one can observe that $S_n^{\ell} -1$ is the minimal telescoper for $f$. Based on this observation, Le~\cite{Le2003} gave a direct
method for computing minimal telescopers for rational functions which avoids the process of item-by-item examination of the order
of the ansatz operators in Zeilberger's algorithm. Motivated by van Hoeij's example in~\cite[Section~3]{VANHOEIJ2025102342}, we partially extend Le's direct method to special hypergeometric terms of the form
\begin{equation}\label{EQ:specialbase}
    H = \frac{q(n,k)}{(mn+\ell k+\alpha)^{\lambda}}\cdot H_0,
\end{equation}
where~$\alpha \in C$,~$\deg_k(q) < \lambda$, $\gcd(q, mn+ \ell k + \alpha) =1$, and~$(mn+\ell k + \alpha)$ is strongly prime with~$K$.
For a nonzero operator $R\in C(n)\langle S_n\rangle$  and a positive integer~$\ell \in \bN$, we can write $R$ as
\begin{equation}\label{EQ:expsep}
    R = R_0 + \cdots + R_{\ell -1},
\end{equation}
with $R_i\in S_n^i\cdot C(n)\langle S_n^\ell\rangle$. This decomposition is called the~\emph{$\ell$-exponent separation} of~$R$, see~\cite[Section~4]{Chen2016}.
\begin{lem}\label{LEM:expsepR}
    Let~$H$ be as in~\eqref{EQ:specialbase} and let $R$ have the~$\ell$-exponent separation as in~\eqref{EQ:expsep}. If $R$ is the minimal prescoper for $H$, then~$R=R_0$.
\end{lem}
\begin{proof}
     Note that any two polynomials in~$\bigl\{ S_n^{i}(mn+\ell k + \alpha)^\lambda \bigr\}{}_{i=0}^{\ell-1} $ are not~$S_k$-equivalent, but for all $j\in \bN$ we have that~$S_n^{j\ell}(mn+\ell k + \alpha)^\lambda$ and $ (mn+\ell k+\alpha)^\lambda$ are $S_k$-equivalent. Then~$\bigl\{ R_i(H)+ \Delta_k(\Omega)\bigr\}{}_{i=0}^{\ell-1} $ is linearly independent over~$C(n)$ modulo~$N$. If $R$ is the minimal prescoper for~$H$, then 
     \[
       R(H) + \Delta_k(\Omega) = \sum_{i=0}^{\ell-1} \bigl( R_i(H) +  \Delta_k(\Omega)\bigr) \in N,
     \]
     which implies that $R_i(H) +  \Delta_k(\Omega) \in N$, i.e., each $R_i$ is a prescoper for~$H$. Since~$N$ is also closed under~$S_n^{-1}$, the trailing coefficient of~$R$ is nonzero, which leads to~$R_0 \neq 0$. For~$i \neq j$, we have~$\ord(R_i) \neq \ord(R_j)$, unless both are zero. We deduce that actually~$R_i = 0$ for each~$i = 1,\ldots,\ell-1$, because otherwise we could find some prescoper~$R_j$ with order less than~$\ord(R)$.
\end{proof}
Using Lemma~\ref{LEM:expsepR}, we now present a recursive algorithm according to the value $\lambda$ for computing the minimal prescoper for $H$ as in~\eqref{EQ:specialbase}. 
Since $S_n^{\ell}S_k^{-m}$ fixes the linear form $(mn+\ell k+\alpha)$, we have
\begin{equation}\label{EQ:defr1}
    h = \frac{S_n^{\ell}S_k^{-m}(H)}{H}=\frac{S_n^{\ell}S_k^{-m}(qH_0)}{qH_0} \in C(n, k).
\end{equation}
Since $\gcd(q, mn+ \ell k + \alpha) =1$ and $(mn+\ell k+\alpha)$ is strongly prime with~$K$, the evaluation of $h$ at
$k = -mn/\ell-\alpha/\ell$, assigned to $r\in C(n)$, is well-defined.

For~$\lambda=1$, we have~$q \in C(n) $. It can be decomposed into
\[
    S_n^\ell(H) \equiv S_n^\ell S_k^{-m}(H) \equiv \biggl(\frac{r\cdot q}{mn+\ell k + \alpha}+ \frac{p'}{v} \biggr) H_0 \bmod \Delta_k(\Omega),
\]
for some~$p' \in W_K$.
Then~$(S_n^\ell -r)\cdot H + \Delta_k(\Omega) \in N$. 
By Lemma~\ref{LEM:expsepR},  we have that $R := S_n^\ell -r$ is the minimal prescoper for~$H$.

For $\lambda>1$, we let~$\widetilde{A}_n := C(n) \langle S_n^\ell\rangle$ which is a subring of~$A_n$ and let $M_{i}$ be the set consisting of the elements 
\[\left(\frac{a}{(mn+\ell k + \alpha)^i} + \frac{p}{v}\right) H_0 + \Delta_k(\Omega) \]
where $ a\in C(n)[k]$ with $\deg_k(a) < i$ and $p \in W_K$. We claim that $M_{i}$ is a
$\widetilde{A}_n$-submodule of~$M$. Indeed, for any~$H_i + \Delta_k(\Omega) \in M_i$ and~$j \in \bN$,
        \[S_n^{j\ell}(H_i) \equiv S_n^{j\ell}S_k^{-jm}(H_i) \equiv \biggl(\frac{a'}{(mn+\ell k + \alpha)^i} + \frac{p'}{v}\biggr) H_0 \bmod \Delta_k(\Omega), \]
        for some~$a' \in C(n)[k]$ with~$\deg_k(a') < i$ and~$p' \in W_K$. By definition, we have~$N\subseteq M_i$ and~$M_{i-1}$ is an $\widetilde{A}_n$-submodule of~$M_i$. 
By the modified Abramov--Petkov\v{s}ek reduction, we can decompose $H$ into
        \[S_n^\ell(H) \equiv S_n^\ell S_k^{-m}(H) \equiv \left( \frac{r\cdot q}{(mn+\ell k + \alpha)^{\lambda}} \right) H_0 + \widetilde H \bmod \Delta_k(\Omega),\]
        where $\widetilde H + \Delta_k(\Omega) \in M_{\lambda-1}$. Then~$ (S_n^{\ell}-r)\cdot H + \Delta_k(\Omega) \in M_{\lambda-1}$. Since
        $R:=S_n^{\ell}-r$ is of order $1$ in~$\widetilde{A}_n$ and $H + \Delta_k(\Omega) \notin M_{\lambda-1}$, it is the minimal annihilator for~$ H +  \Delta_k(\Omega)  + M_{\lambda -1}\in M_{\lambda}/M_{\lambda-1}$. We can recursively compute the minimal prescoper $\widetilde{L}$ for $\widetilde H$. By Lemma~\ref{LEM:rightfactor}, we have $\widetilde{L}\cdot R$ is the minimal prescoper for~$H$.

The following example, sent to us by Hui Huang, indicates that the above method outperforms
the existing codes for Zeilberger's algorithm in Maple and the reduction-based method in~\cite{ChenHuangKauersLi}.
\begin{exam}
    Consider the hypergeometric term
    \[H = \frac{1}{2n+k} H_0 \quad \text{with $H_0 = 
    \frac{\binom{5n}{3k}^2}{\binom{n}{k}}$}.\]
    Then the shift-quotient with respect to $k$ is 
    \[K=\frac{S_k(H_0)}{H_0} = \frac{(3k-5n)^2(3k-5n+1)^2(3k-5n+2)^2}{9(n-k)(k+1)(3k+1)^2(3k+2)^2},\]
    which is already shift-reduced in $k$.
    Let $v$ be the denominator of $K$ and 
    \[N = \left\lbrace \frac{p}{v} H_0 + \Delta_k(\Omega) \ \middle| \ p \in W_K \subset \bQ(n)[k] \right\rbrace .\]
    Observe that $H \notin N$. Evaluating ${S_n S_k^{-2}(H)}/{H}$ at $k=-2n$ yields
    \begin{align*}
    r = \frac{3(3n+1)(3n+2)\prod_{i=1}^{5}(5n+i)^2 \prod_{i=0}^{5}(6n+i)^2}{2n(2n+1)\prod_{i=1}^{11}(11n+i)^2}.
    \end{align*}
    Then $R = S_n-r$ is the minimal prescoper for $H$. It remains to compute the minimal telescoper for 
    \[\widetilde{H} := (S_n-r)\cdot H,   \]
    which is of order 6. It takes 13 seconds on a Dell Optiplex 7090 (CPU 3.70GHz, RAM 128G) with the reduction-based method in~\cite{ChenHuangKauersLi}, compared with 31 seconds with the Maple code for Zeilberger's algorithm.
\end{exam}

Note that in this example, $H_0$~is not defined for all integers. This is a bit uncommon, but it is 
not so uncommon that the certificates have poles at some integer points. 
This also happens in some of the examples discussed below. 
Algorithms and theory for creative telescoping typically ignore this issue and leave it to the user to 
check that everything makes sense. Noteworthy exceptions include the careful study of Abramov and 
Petkov\v sek~\cite{abramov06,abramov05} as well as the approach of Bostan, Lairez, and Salvy~\cite{BostanLairezSalvy16}.

\subsection{Automorphisms of the kernel submodule}

In his paper~\cite{VANHOEIJ2025102342}, van Hoeij presents examples in which a symmetry of a summation problem
translates into an automorphism of the submodule~$N$. The eigenspaces of the automorphism give
rise to a decomposition of $N$ into submodules, and this decomposition explains why the minimal
telescoper is not the minimal annihilating operator of the sum.

Automorphisms of $N$ can be found algorithmically. By Lemma~\ref{THM:submoduleN}, the $A_n$-module
$N$ has a finite dimension as $C(n)$-vector space. Let $\lbrace v_1, \ldots, v_d \rbrace $ be a vector space basis.
Any $A_n$-automorphism $\phi \colon N \rightarrow N$ is in particular a $C(n)$-linear map. As such, it 
can be written in the form 
\begin{equation}\label{EQ:defPhi}
    \left(
        \begin{matrix}
            \phi(v_1) \\ \vdots \\ \phi(v_d)
        \end{matrix}
    \right) = \Phi 
    \left(
        \begin{matrix}
            v_1 \\ \vdots \\ v_d
        \end{matrix}
    \right),
\end{equation}
for a certain matrix $\Phi \in C(n)^{d\times d}$. The requirement for a linear map to be an $A_n$-module
automorphism is that it is invertible and compatible with the shift. If $\Sigma \in C(n)^{d \times d}$ is defined by
\begin{equation}\label{EQ:defS}
    \left(
        \begin{matrix}
            S_n(v_1) \\ \vdots \\ S_n(v_d)
        \end{matrix}
    \right) = \Sigma 
    \left(
        \begin{matrix}
            v_1 \\ \vdots \\ v_d
        \end{matrix}
    \right), 
\end{equation}
then the latter requirement means that the commutation rule $\Sigma\Phi=S_n(\Phi)\Sigma$ must hold.

In order to find automorphisms, we can therefore make an ansatz with undetermined coefficients for the entries
of~$\Phi$. The requirement $\Sigma\Phi=S_n(\Phi)\Sigma$ leads to a coupled system of linear recurrence equations
for the undetermined coefficients. This system can be solved using the command \texttt{SolveCoupledSystem}
of Koutschan's Mathematica package \texttt{HolonomicFunctions}~\cite{Koutschan20210}. The result is a $C$-linear subspace
of $C(n)^{d\times d}$. Automorphisms correspond to all the matrices in this space whose determinant is nonzero.

\begin{exam}
    For the hypergeometric term~$H_0 := \binom{n}{2k}^2$, the kernel module~$N$ computed by the modified Abramov--Petkov\v{s}ek reduction is a $C(n)$-vector space of dimension~$3$, given by the following basis:
    \[
      \bigg\lbrace \frac{k^i}{4(2k+1)^2(k+1)^2} H_0 + \Delta_k(\Omega)
      \mathrel{\bigg|} i = 0,1,2 \bigg\rbrace.
    \]
    The matrix~$\Sigma\in C(n)^{3\times3}$ is determined by~$N$, but it is too large to display it here.
    We make an ansatz for $\Phi := (\phi_{i,j})_{1\leq i,j\leq 3}$ with undetermined entries~$\phi_{i,j}$. Then the condition $\Sigma\Phi=S_n(\Phi)\Sigma$ yields a $9\times9$ coupled first-order linear system of difference equations, whose rational solutions are computed with the command \texttt{SolveCoupledSystem}.
    It returns a two-dimensional solution space over the constant field~$C$, which is spanned by the identity matrix~$I$ and by the matrix
    \[
      \Psi=\frac{
      \begin{psmallmatrix}
        12 n^3+16 n+64
        & -64 n^2+32 n+192
        & 128 n+192 \\[2pt]
        4 n^4-2 n^3+4 n^2-8 n-48 
        & -20 n^3+16 n^2-16 n-128 
        & 32 n^2-16 n-96 \\[2pt]
        n^5-2 n^4+4 n^3+32 
        & -4 n^4+16 n^3-16 n^2-32 n+64 
        & 4 n^3-40 n^2-48 n+32
      \end{psmallmatrix}}{4(n+2)^3}.
    \]
    The matrix~$\Psi$ corresponds to the automorphism $(n,k)\to(n,k+1/2)$, and it satisfies $\Psi^2=I$, as expected. By inspecting the symmetry of~$H_0$, one could anticipate the existence of another automorphism, namely $(n,k)\to(n,n/2-k)$. However, it turns out that this map is not compatible with $S_n$ and hence is not an $A_n$-module automorphism.
\end{exam}

\subsection{Zero-sum submodules}\label{SUBSECT:zerosum}
The submodule approach introduced by van Hoeij~\cite{VANHOEIJ2025102342} can not only speed-up the computation of minimal telescopers, but also explain (by examples) why the minimal telescoper for a hypergeometric sum may not be its minimal recurrence. 
The explanation of the non-minimality phenomenon by anti-symmetry has been given in~\cite{Chyzak2000ANM, Paule1994,Paule2021,PauleRiese1997,PauleSchorn1995} that leads to the method of creative symmetrizing~\cite{Le2002CSM}.  
A concrete example is the identity
\[\sum_{k=0}^{2n+1} (-1)^k\binom{2n+1}{k}^2 = 0. \]
The summand $H := (-1)^k\binom{2n+1}{k}^2$ satisfies the anti-symmetry relation 
\[H(n, k) = - H(n, 2n+1-k).\]
So summing $H$ for $k$ from $0$ to $2n+1$ leads to zero.  The minimal telescoper for $H$ is the first-order operator~$S_n + 8(n+1)/(2n+3)$, but the minimal recurrence for the above vanishing sum is
any nonzero element of $C(n)$.

As a research question, van Hoeij~\cite[Section~7]{VANHOEIJ2025102342} proposed  to study the zero-sum submodules, especially how to detect and find such submodules. We call $Z\subseteq N$ a zero-sum submodule if it only contains terms whose summation with respect to~$k$ gives~$0$. Note that every operator~$T$ with $T(\overline H)\in Z$ is then an annihilating operator of $\sum_kH$, but not necessarily a telescoper.

The following two examples show how the techniques from the previous sections, especially Nicole’s lemma, can be used to construct zero-sum submodules and explain the non-minimality phenomenon. In the first example, we find that the minimal prescoper~$R$ maps $\overline H$ not only into~$N$ but even into~$Z$. It is therefore an annihilator of the sum. However, since $R(\overline H)\neq0\in M$, it is not a telescoper. In the second example, the minimal prescoper is $R=1$. Nevertheless, the minimal telescoper is not the minimal annihilator of the sum because it turns out that there is an operator~$T$ with $T(\overline H)\in Z$ but $T(\overline H)\neq0$.

\begin{exam}
The minimal telescoper for the hypergeometric term
\[
  H :=(-1)^k \binom{3n+1}{k} \binom{3n-k}{n}^3
\]
is of order~$2$, which is not the minimal recurrence satisfied by the sum
\[
  \sum_{k=-\infty}^{+\infty} H(n, k) = 1.
\]
To explain this non-minimality, we let $H_0 = (k-3n-1)H$ and let
\[
  K := \frac{S_k(H_0)}{H_0} = \frac{(k-2n)^3}{(k+1)(k-3n)^2}=: \frac{u}{v}.
\]
Then the algorithm in Section~\ref{SUBSECT:direct} can compute the minimal prescoper $R=S_n -1$ for $H$ so that $R(H) + \Delta_k(\Omega)$ is in the submodule
\[
  N := \Bigl\{ \frac{p}{v}\cdot H_0 + \Delta_k(\Omega) \mathrel{\Big|} p \in W_K \Bigr\},
\]
where $W_k$ has a $\bQ(n)$-basis $\lbrace 1, k^3 \rbrace$. 
We now use Nicole's lemma in Section~\ref{SUBSECT:ratsum} to show that for all $p\in \bQ(n)[k]$ with $\deg_k(p) \leq 2 $, we have the vanishing-sum identity 
\[
  \sum_{k=-\infty}^{+\infty} \frac{p}{v}\cdot H_0 =0,
  \quad \text{where $n\geq 1$}.
\]
Similar to Examples~\ref{EX:euler} and \ref{EX:vanisum}, we consider the rational function
\begin{align*}
    f & = \frac{P}{Q} = \frac{p(n,-x)(3n+1)!(x+3n-1)^2\cdots (x+2n+1)^2}{(n!)^3(x-1)x(x+1)\cdots(x+2n)} \\
      & = \sum_{k=-1}^{2n} \frac{\alpha_k}{x+k}.
\end{align*}
    Since $Q$ is squarefree,  Lagrange's residue formula implies that
\[
        \alpha_k  = \frac{p(n,k)(3n+1)!(3n-k-1)^2\cdots(2n-k+1)^2}{(n!)^3(-k-1)(-k)(-k+1)\cdots(-k+2n)} = \frac{p}{v}\cdot H_0.
\]
    By Lemma~\ref{THM:nicole}, $f$ is summable in $\bC(x)$ since $\deg_k(p) \leq 2$.  Then the above vanishing-sum identity holds. By this identity, we have a zero-sum submodule
    \[Z := \left \lbrace \frac{p}{v}\cdot H_0 + \Delta_k(\Omega) \ \middle| \ p \in W_K\,  \text{with}\,  \deg_k(p) = 0 \right\rbrace.\]
Applying the prescoper $R = S_n - 1$  to $H$ yields
    \[S_n(H)-H \equiv \frac{37n^7+96n^6+81n^5+22n^4}{8(n+1)^3(9n^2+10n+3)} \frac{H_0}{v}\bmod \Delta_k(\Omega).\]
   So $S_n(H)-H + \Delta_k(\Omega)\in Z$ which contributes zero to the sum.  Then $S_n - 1$ is the minimal annihilator for the sum $\sum_{k=-\infty}^{+\infty} H(n, k)$. 
 \end{exam}

\begin{exam}
We now explain why minimal telescopers overshoot in the following combinatorial identity
\[\sum_{k=0}^n(-1)^k \binom{n}{k} \binom{3k}{n}  = (-3)^n.\]
This is a special case of the identity  in~\cite[Section~4.3]{PauleSchorn1995}
which was originally used to show the non-minimality phenomenon
with explanations in~\cite{BostanLairezSalvy16, Paule2021}.
The minimal telescoper for the summand~$H:= (-1)^k \binom{n}{k} \binom{3k}{n} $ is 
    \[S_n^2 + \frac{3(5n+7)}{2(2n+3)}S_n+\frac{9(n+1)}{2(2n+3)},\]
but this is not the minimal recurrence $S_n+3$ satisfied by the sum. In this example, we let $H_0 = H$ and 
    \[K := \frac{S_k(H_0)}{H_0} = \frac{3(k-n)(3k+1)(3k+2)}{(3k-n+1)(3k-n+2)(3k-n+3)} =: \frac{u}{v}.\]
The corresponding kernel submodule is
\[
  N := \Bigl\{ \frac{p}{v}\cdot H_0 + \Delta_k(\Omega) \mathrel{\Big|} p \in W_K \Bigr\},
\]
where $W_k$ has a $\bQ(n)$-basis $\{ 1, k^3\}$. Since $H+ \Delta_k(\Omega) \in N$, the minimal prescoper of~$H$ is $R=1$. Similar to the previous example, considering the rational function 
\[f = \frac{p(n,-x)(-3x)(-3x-1)\cdots(-3x-n+4)}{x(x+1)\cdots(x+n)}\]
yields the vanishing-sum identity
\[
  \sum_{k=0}^n \frac{p}{v} \cdot H_0 = 0,
  \quad \text{where $n\geq 3$},
\]
for all  $p\in \bQ(n)[k]$ with $\deg_k(p) \leq 2 $.
So we obtain the zero-sum submodule
    \[Z := \left \lbrace \frac{p}{v}\cdot H_0 + \Delta_k(\Omega) \ \middle| \ p \in W_K \,  \text{with}\, \deg_k(p) = 0 \right\rbrace.\]
    We can verify that $Z$ is closed under any operator in $A_n$. In fact,
    \[S_n\left(\frac{H_0}{v}\right) \equiv \frac{-9n^3-21n^2+36n+84}{2(n+2)(2n+5)(3n+4)}\frac{H_0}{v} \bmod \Delta_k(\Omega).\]
    The remaining task is to find an operator $T \in \bQ(n)\langle S_n \rangle$ such that $T(H) + \Delta_k(\Omega) \in Z$.
The modified Abramov--Petkov\v{s}ek reduction decomposes $H_0$ and $S_n(H_0)$ as 
    \begin{align*}
    H_0 &\equiv \frac{81k^3n-n^4+108k^3+4n^3-12n^2+12n+18}{(3n+4)\cdot v} H_0 \\
      &\bmod \Delta_k(\Omega).\\
      S_n(H_0) &\equiv \frac{-243k^3n+n^4-324k^3-9n^3+41n^2-42n-54}{(3n+4)\cdot v} H_0 \\
    & \bmod \Delta_k(\Omega).
    \end{align*}
Note that $T = S_n +3$ brings $H_0$ into the zero-sum submodule~$Z$. Therefore, $T$ annihilates the sum, and since the sum evaluates to~$(-3)^n$, we find that~$T$ is actually its minimal annihilator.
\end{exam}
\bigskip 

\noindent {\bf Acknowledgements.} 
We thank Hui Huang for sharing her example and making her Maple codes available to us. We also thank Huajun Bian and Yiman Gao for many discussions.
We also thank the anonymous referees as well as Alin Bostan and Peter Paule for pointing us to additional relevant references.

\bibliographystyle{plain}


\end{document}